\documentclass[12pt,twoside,reqno]{amsart}
\usepackage[colorlinks=true,citecolor=blue]{hyperref}
\usepackage{mathptmx, amsmath, amssymb, amsfonts, amsthm, mathptmx, enumerate, color,mathrsfs}
\usepackage{graphicx}

\usepackage{multirow}
\usepackage{epstopdf}
\usepackage{multicol}
\usepackage{algorithm}
\usepackage{algorithmic}
\usepackage{epstopdf}

\providecommand{\R}{\mathbb{R}}
\providecommand{\N}{\mathbb{N}}

\providecommand{\p}{\mathbb{P}}
\providecommand{\Q}{\mathbb{Q}}

\providecommand{\A}{\mathcal{A}}

\providecommand{\LL}{\mathcal{L}}

\providecommand{\M}{\mathcal{M}}

\providecommand{\HH}{\mathcal{H}}
\providecommand{\Norm}{\mathcal{N}}

\providecommand{\E}{\mathbb{E}}

\providecommand{\half}{\ensuremath{\frac{1}{2}}}

\newcommand*{\defeq}{\mathrel{\vcenter{\baselineskip0.5ex \lineskiplimit0pt
                     \hbox{\scriptsize.}\hbox{\scriptsize.}}}%
                     =}

\newtheorem{theorem}{Theorem}[section]
\newtheorem{lemma}[theorem]{Lemma}

\theoremstyle{definition}

\newtheorem{remark}[theorem]{Remark}
\numberwithin{equation}{section}

\begin{document}
\setcounter{page}{1}

\vspace*{2.0cm}
\title[S Shaped Utility Maximisation with a Random Reference Point]
{Deep Learning Methods for S Shaped Utility Maximisation with a Random Reference Point}
\author[A. Davey, H. Zheng]{ Ashley Davey$^{1}$, Harry Zheng$^{2, *}$}
\maketitle
\vspace*{-0.6cm}

\begin{center}
{\footnotesize

$^1$Department of Mathematics, Imperial College, London SW7 2BZ,\, UK\\
$^2$Department of Mathematics, Imperial College, London SW7 2BZ,\, UK\\
This paper is dedicated to Professor Richard B. Vinter on the occasion of his 75th birthday. 

}\end{center}

\vskip 4mm {\footnotesize \noindent {\bf Abstract.}
We consider the portfolio optimisation problem where the terminal function is an S-shaped utility applied at the difference between the wealth and a random benchmark process. We develop several numerical methods for solving the problem using deep learning and duality methods. We use  deep learning methods to solve the associated Hamilton-Jacobi-Bellman  equation for both the primal and dual problems, and the adjoint equation arising from the stochastic maximum principle. We compare the solution of this non-concave problem to that of concavified utility,  a random function depending on the benchmark,  in both complete and incomplete markets. We give some numerical results for power and log utilities to show the accuracy of the suggested algorithms.


 }

\renewcommand{\thefootnote}{}
\footnotetext{ $^*$Corresponding author.
\par
E-mail address: ashley.davey18@imperial.ac.uk (A. Davey), h.zheng@imperial.ac.uk (H.Zheng).
}

\section{Introduction}

In this paper we apply the deep learning method to solve utility maximisation problem with a  benchmark that is  a non-negative random variable. The agent wishes to maximise an S-shaped utility, a function which is concave on $(0,\infty)$ and convex on $(\infty, 0)$, applied at the difference between their terminal wealth and the benchmark. This problem is non-standard and difficult to solve due to the non-concavity of the problem. We suggest several numerical methods for solving the problem utilising the concave envelope of the utility and the dual control method.

 The random benchmark may represent an asset or portfolio of assets (that could be partly non-tradable) whose payoff we try to replicate. Unlike the deterministic case, the comparison to a random variable allows us to set a reference for utility maximisation that accounts for natural fluctuations in the market. In particular, the benchmark may be proportional to the stock market index. This makes sense intuitively, as the targets for an investor may  increase or decrease over time when the market is favourable or challenging. The benchmark may also be affected by factors which cannot be traded, leading to incompleteness of the problem.

The treatment of utility maximisation problems is standard in the literature, see \cite{17} for exposition of the dynamic programming principle (DPP) and 
the stochastic maximum principle (SMP), and the corresponding Hamilton-Jacobi-Bellman (HJB) equation and the adjoint backward stochastic differential equation (BSDE). The duality approach has been regularly applied to this problem in both complete and incomplete markets, see \cite{4, 13, 14}. The control problem we wish to solve can be seen as an extension of regular utility maximisation problem, where the utility function is random, and also non-concave in the state variable. 

Having randomness in the dynamics and gains functions leads to a SMP involving Malliavin derivatives \cite{16}, but as long as the random utility has sufficient integrability, the adjoint BSDE admits a solution, we may apply the SMP in the same way as the standard approach of, say, \cite{21}. The duality approach has been considered for a random utility that arises from a random endowment at terminal time, see \cite{5,10}.

The introduction of an S-shaped structure allows us to compare two random variables. However, it is a tricky function to deal with as it is non-concave, and often this utility will have infinite derivative when the variables are equal. Non-concave and potentially discontinuous utilities have been treated in the literature, see \cite{1,2,19}. Maximisation of an S-shaped utility compared to a benchmark has been considered in \cite{8, 20}, but unlike our setting the benchmark is deterministic.

A key notion for studying non-concave problems is the so called concavification principle. This principle states that we can instead maximise the concave envelope of the utility function and get the same solution as if we solved the non-concave problem directly. The concavified problem is solved with  standard approaches. In the deterministic benchmark, complete market and no control constraint case, the concavification principle is known to hold, see \cite{19}. The value function is a smooth solution to the HJB equation in the Markovian setting, see \cite{1}, and is the conjugate function of the dual value function that has a simple integral representation and may have a closed-form   expression for some utility functions, see \cite{8}. The optimal control drives the wealth to points where the utility matches its concave envelope. In our case, that is  the set $\{0\} \cup [Z, \infty)$ for some $Z > 0$ greater than the benchmark. If the wealth is in the range $(0, Z)$ before the terminal time, the optimal control involves considerable risk and creates two possibilities: the terminal wealth either beats the benchmark by a certain amount or becomes 0. However, if in some cases this optimal policy cannot be attained due to control constraints or incompleteness, the value function is not concave and the concavification principle does not apply, see \cite{6}.

We use the deep learning method to solve the primal, concavified and dual control problems arising in our setting. A key benefit of deep learning methods is the scalability with dimension \cite{9}. We introduce various PDEs, which can be solved numerically using a Physics Informed Neural Network (PINN), introduced in \cite{18}. The solution function is taken to be a neural network in state and time, and the PDE and terminal condition are evaluated and used as loss functions to find the optimal neural network. To solve the adjoint BSDE associated to the SMP, we use a combination of the deep solver of \cite{11}, the deep BSDE method in \cite{12} and the deep SMP method in \cite{7}. In this method the driving process is taken as a neural network and the  wealth process as input. The adjoint equation is simulated using Monte Carlo (in the forward direction), with the driver process determined using the SMP. The terminal condition is used as a loss function for the initial value of the adjoint process. We also use a neural network for the control process, as in \cite{7}, and optimise the control process  using  the gains function as a loss function, accelerated with the terminal condition for the adjoint equation. 

The remainder of the paper is outlined as follows. Section \ref{sec_2} formulates the problem and defines the concavified problem. Section \ref{sec_3} considers the case with a scalable power utility and reduces the dimension of the problem, and the PDE method is used to solve the reduced problem and its  concavified and dual problems. Section \ref{sec_4} describes the PDE and SMP methods for problems with general utilities. Section \ref{sec_5} presents numerical results comparing all the methods for a range of problems, and Section \ref{sec_6} concludes.

\section{Problem Formulation} \label{sec_2}

Assume  $(\Omega, \mathcal{F},   \{\mathcal{F}_t\},  \p)$ is a filtered
probability space with $\{\mathcal{F}_t\}$ the natural filtration
generated by two standard Brownian motions $W, W^R$ with correlation $\rho \in [-1, 1]$,  augmented with all $\mathbb{P}$-null sets.  Assume  the investment period is $[0,T]$ and
the risk-free asset price $S_0$ and the risky asset price $S$ satisfy the following stochastic differential equations (SDEs), for $t \in [0,T]$,
\begin{align*}
dS^0_t & = \alpha S^0_t dt, \quad 
dS_t  = \mu S_tdt + \sigma S_t dW_t,
\end{align*}
where $\alpha, \mu, \sigma$ are positive constant, representing the riskless interest rate, stock growth and  volatility rates. 
An investor invests a proportion $\pi \in \R$ of the wealth $X$ in $S$ and the remainder in $S^0$. The wealth process $X$ satisfies, for $t \in [0,T]$,
$$ dX_t  =  \left(\alpha X_t + \pi_t X_t \theta \sigma \right) dt + X_t \pi_t \sigma dW_t
$$
with initial wealth $X_0=x_0> 0$, where $\theta =(\mu-\alpha)/\sigma$ is the market price of risk and $\pi\in \HH^2_0(\R)$, the set of all square integrable and progressively measurable processes on $[0,T]$ taking values in $\R$. The investor wants to maximise an S-shaped utility of the excess of the wealth $X$ compared to some positive benchmark process $R$ satisfying,  for $t \in [0,T]$,
$$dR_t  =  a R_t dt + b R_t dW^R_t
$$
with initial value $R_0=r_0>0$, where $a,b$ are nonnegative constants. 
The utility function $U \colon \R \to \R$ is given by, for $z \in \R$,
\begin{align*}
U(z) & \defeq 
\begin{cases}
U_1(z) & z \geq 0, \\
-U_2(-z) & z < 0,
\end{cases} 
\end{align*} 
where $U_1, U_2 \colon [0, \infty) \to \R$ are strictly increasing, strictly concave, continuously differentiable with $U_1(0) = U_2(0) = 0$. The value function is defined by, for $(t, x, r) \in [0,T] \times (0, \infty)^2$,
\begin{align}\label{eq_value}
v(t, x, r) &  \defeq \sup_{\pi \in \HH^2_t(\R)}\E\left[ U(X_T - R_T)|X_t = x, R_t = r \right].
\end{align}

The combination of two utilities leads to an S-shaped function $U$. To solve this problem we use the concave envelope and Fenchel Legendre transforms, applied to $x$ variable, that is, the concave envelope is the smallest concave function (in $x$)  $\bar{U}(x, r)$ that is greater than or equal to  $ U(x - r)$. By \cite{3} there exists a function $ \eta \colon [0, \infty) \to [0, \infty)$ such that $\eta(r) \geq r$ for all $r \in [0,\infty)$ and
\begin{align}
U_1(\eta(r) - r) + U_2(r) - \eta(r) U_1'(\eta(r) - r) = 0,
\label{eq_tangent}
\end{align}
and the concave envelope is given by
\begin{align} \label{eq_def_u_bar}
\bar{U}(x, r) & = 
\begin{cases}
-U_2(r) + x U_1'(\eta(r) - r) & x <  \eta(r), \\
U_1(x - r) & x \geq  \eta(r) .
\end{cases}
\end{align}
Note $\bar U(\cdot, r)$ is concave, increasing, and $C^1$ for fixed $r\geq 0$.
Define the value function of the concavified problem by, for   $(t, x, r) \in [0,T] \times (0, \infty)^2$,
\begin{align}\label{eq_value_conc}
\bar{v}(t, x, r) &  = \sup_{\pi \in \HH^2_t(\R)}\E\left[ \bar{U}(X_T, R_T)|X_t = x, R_t = r \right].
\end{align}

The Fenchel Legendre transform of $U$ (in $x$) is given by, for $(y, r) \in (0,\infty)^2$,
\[\tilde{U}(y, r) \defeq \sup_{x > 0} \left\{ U(x - r) - xy \right\} = \sup_{x > 0} \left\{ \bar{U}(x, r) - xy \right\}.\]
Substituting (\ref{eq_def_u_bar}) into the above definition, we have 
\begin{align} \label{eq_dual_utility}
\tilde{U}(y, r) & = \begin{cases}
\tilde{U}_1(y) - ry & 0 < y \leq  U_1'( \eta(r) - r), \\
-U_2(r) & y > U_1'( \eta(r) - r),
\end{cases}
\end{align}
where $\tilde{U}_1$ is the Fenchel Legendre transform of $U_1$, and $ \eta$ is determined by (\ref{eq_tangent}).

It is clear that function $\bar{v}$ is concave in $x$,  $\bar{v} \geq v$ for all for $(x, r) \in (0,\infty)^2$, and   $v(T, x, r) \ne \bar{v}(T, x, r)$ on the non-concavity region $ C \defeq \{(x, r) \in [0,\infty)^2 \colon U(x - r) \neq \bar{U}(x, r)\}$. However, under some circumstances equality can hold at terminal time. This is due to the optimal state process $X^*$ satisfying $\p((X^*_T, R_T) \in C) = 0$. The optimal control `steers' the state process either to a region where $U$ is concave or to $0$, points at which $\bar U$ equals $U$. 

\begin{theorem}\label{thm_concavification} \cite[Theorem 1]{3}.  Suppose the market is complete ($|\rho| = 1$) and the function
\[\psi(\lambda) = -\E\left[\zeta_T \partial_y \tilde{U}(\lambda \zeta_T, R_T)\right] \]
satisfies $\psi(\lambda) < \infty$ for all $\lambda > 0$, where
$\zeta_t = e^{-(r + \half \theta^2) T - \theta W_T}$. 
Then there exists $\lambda^* > 0$ such that $\psi(\lambda^*) = x_0$ and
\[X_T^* = - \partial_y \tilde{U}\left(\lambda^* \zeta_T, R_T\right).\]
In particular, we have $\p((X^*_T, R_T) \in C) = 0$. 
\end{theorem}

Theorem \ref{thm_concavification} explicitly describes the behaviour of the optimal state process $X^*$ in terms of the state price density $\zeta$. In particular, given the structure of the dual utility $\tilde{U}$  in (\ref{eq_dual_utility}), we see that there exists a threshold random variable  $\zeta^* \defeq  \frac{1}{\lambda^*}U_1'( \eta(R_T) - R_T)$ such that $X^*_T = 0$ if and only if $\zeta_T \geq \zeta^*$.

\begin{remark}
This result comes from \cite{3} in which  a utility of the form $U(A(X_T - R_T)^+ + B)$ is considered for some $A, B > 0$, so the utility is constant when $X_T < R_T$, with $U_2 \equiv -U(B)$. However, the proof still holds for a general S-shaped utility with $U(-R_T) = -U_2(R_T) > -\infty$ almost surely.
\end{remark}

Given this concavification principle we have $v(t, x, r) = \bar{v}(t, x, r)$ for all $t < T$
in the complete market setting, and we may proceed to solve the concavified problem.

\section{PINN for Power Utility} \label{sec_3}

In this section we consider the utility, for $z \in \R$,
\begin{align*}
U(z) = \begin{cases}
z^p & z \geq 0, \\
- K |z|^p & z < 0,
\end{cases}
\end{align*}
for some $K \geq 0$ and $p \in (0,1)$.  In this setting, the dual utility (\ref{eq_dual_utility}) becomes, for $y, r > 0$,
\begin{align*} 
\tilde{U}(y, r) & = \begin{cases}
(1 - p)\left(\frac{y}{p}\right)^{\frac{p}{p-1}} - ry & 0 < y \leq  p({\color{blue} \eta}(r) - r)^{p - 1}, \\
-Kr^p & y > p({\color{blue} \eta}(r) - r)^{p - 1}.
\end{cases}
\end{align*}

\subsection{Problem Reduction}

 Define a non-negative process $Z_t = \frac{X_t}{R_t}$ for $t \in [0,T]$.  Then $Z_t$ satisfies the following SDE:
\begin{align*}
dZ_t 
& = Z_t \left(\alpha - a + b^2 + \pi_t  \sigma \left(\theta - \rho b\right)\right) dt + Z_t \pi \sigma dW_t - Z_t b dW^R_t. 
\end{align*}
 Define an exponential martingale
$ F_t = e^{pb W^R_t - \half p^2 b ^2 t}$ for $t \in [0,T]$
and a new measure $\Q$ by $\frac{d\Q}{d \p} =  F_{T}$. Then by Girsanov's theorem, the processes
\begin{align*}
\tilde{W}_t &  = W_t - \rho p b t, \quad
\tilde{W}^R_t = W^R_t - p b t,
\end{align*}
are $\Q$ Brownian motions with correlation $\rho$. SDE for $Z$ can be written equivalently as
$$
dZ_t  =  Z_t \left(\alpha_0  + \pi_t \sigma\bar{\theta}\right) dt + Z_t \pi \sigma d\tilde{W}_t - Z_t b d\tilde{W}^R_t, 
$$
where $\alpha_0 \defeq \alpha - a - b^2 (p - 1)$ and $\bar{\theta} \defeq \theta + \rho b ( p - 1)$. 
Furthermore, by Bayes' formula (eg \cite[Prop. 1.2.12]{17}), we have, for $(t, x, r) \in [0,T] \times [0, \infty) \times (0, \infty)$, 
\begin{align} \begin{split} \label{eq_value_conversion}
v(t, x, r) &= \sup_{\pi \in \HH^2_t(\R)} \E\left[ R_T^p U(Z_T - 1) |Z_t=z\right]\\
&  = \sup_{\pi \in \HH^2_t(\R)} \E^\Q\left[ \frac{F_t R_T^p}{ F_T} U(Z_T - 1)|Z_t=z \right]\\
&= r^p H_{T - t} g(t,z),
\end{split}
 \end{align}
 where $z=x/r$, $H_s \defeq e^{p\left(a + \half b^2 (p - 1)\right)s}$ for $0 \leq s \leq T$, and
\begin{align} \label{eq_value_new}
g(t, z) \defeq  \sup_{\pi \in \HH^2_t(\R)} \E^{\Q}\left[ U(Z_T - 1) \middle| Z_t = z\right].
 \end{align}
The HJB equation associated to the  problem (\ref{eq_value_new}) is 
\begin{equation} \label{eq_hjb_reduced}
 \partial_t g + \sup_{\pi \in \R}\left\{\partial_z g z \left(\alpha_0  + \pi \sigma\bar{\theta}\right) + \half \partial_{zz} gz^2 (\pi^2 \sigma^2 -  2 \rho b \sigma \pi + b^2)\right\}=0,
\end{equation}
with the terminal condition $g(T, z) = U(z - 1)$ for $z \in [0, \infty)$.  
The candidate optimal control is given by,  for $(t, z) \in [0,T] \times [0, \infty)$,
\begin{align}\label{eq_opt_control_reduced}
\pi^*(t, z) = - \frac{\partial_z g(t, z)  \bar{\theta}  }{\partial_{zz} g(t, z)z \sigma} +   \frac{\rho b}{\sigma},
\end{align}
which leads to a nonlinear PDE
\[ \partial_t g + \partial_z g z \alpha_0  + \half \partial_{zz} gz^2 b^2 - \half \frac{\left(\partial_z g  \bar{\theta} -  \partial_{zz} gz \rho b  \right)^2}{\partial_{zz} gz \sigma}=0.\]

Problem (\ref{eq_value_new}) is simpler as the reference is now a constant. However the function $U$ is still a non-concave function for $z\geq 0$ and non-differentiable at $z=1$. Define, for $(t, z) \in [0,T] \times [0, \infty)$,
$$
\bar{g}(t, z)  \defeq \sup_{\pi \in \HH^2_t(\R)} \E^{\Q}\left[ \bar{U}(Z_T, 1) \middle| Z_t = z\right],
$$
 where $\bar U(z,1)$ is the concave envelope of $U(z-1)$ for $z\geq 0$. 
The concavified value function $\bar g$ satisfies the same equation (\ref{eq_hjb_reduced}) with the concave terminal condition
$\bar{g}(T, z)  = \bar{U}(z, 1)$. In both cases we can use the fact that $Z$ is a geometric Brownian motion for any $\pi$, with $Z_T = 0$ if and only if $Z_t = 0$ for any $t < T$, so we can use the boundary condition
\[g(t, 0) = \bar{g}(t, 0) = U(-1) = -K, \quad 0 \leq t \leq T.\]

If $|\rho| = 1$, it is known that  $g$ coincides with $\bar{g}$ on $[0,T) \times (0,\infty)$, see   Section \ref{sec_solution}. If $|\rho|<1$, it is not yet clear if there exists a duality gap between $g$ and $\bar g$.

\subsection{Dual Problem} 

Due to incompleteness arising from  Brownian motion $W^R$, the dual problem contains a dynamic dual control and is  a true control problem. We determine the dual problem explicitly. Firstly, note that the process $Z$ is non-negative as  $X$ and $R$ are non-negative, with $R$ strictly positive. The dual process $Y_t$ has the form, for  $0 \leq t \leq T$, 
\begin{align*}
dY_t = -Y_t \gamma_t dt -Y_t \beta_t d\tilde{W}_t - Y_t \xi_t d\tilde{W}^R_t,
\end{align*}
where $\gamma, \beta, \xi$ and $Y_0=y$ are to be determined. We have
\begin{align*}
d(YZ)_t & = Y_t Z_t \left(\alpha_0 - \gamma_t + b(\xi_t + \rho \beta_t) + \pi_t \sigma \left(\bar{\theta} - (\rho \xi_t + \beta_t)\right)\right) dt \\
& {} + Y_tZ_t (\pi_t\sigma - \beta_t) d\tilde{W}_t - Y_tZ_t(\xi_t + b)d\tilde{W}^R_t.
\end{align*}
To ensure that $YZ$ is a $\Q$-supermartingale for all $\pi \in \HH^2_0(\R)$, we must have
$$
 \alpha_0 - \gamma_t + b(\xi_t + \rho \beta_t) + \pi_t \sigma \left(\bar{\theta} - (\rho \xi_t + \beta_t)\right) \leq 0, 
$$
which implies
$\beta_t  = \bar{\theta} - \rho \xi_t$.
We may choose a minimal $\gamma_t$ as 
$\gamma_t  = \alpha_0 + b(1 - \rho^2)\xi_t + b \rho \bar{\theta} = \bar{\alpha} + b(1 - \rho^2)\xi_t$,
where $\bar{\alpha} \defeq \alpha - a + \rho b \theta$. 
This gives us the dual process, for $0 \leq t \leq T$,
\begin{align*}
dY_t = -Y_t \left(\bar{\alpha} + b(1 - \rho^2)\xi_t \right) dt - Y_t \left( \bar{\theta} - \rho \xi_t \right) d\tilde{W}_t - Y_t \xi_t d\tilde{W}_t^R
\end{align*}
with dual control process $\xi$. We then have, for any $\pi, \xi, Y_t=y, Z_t=z$, $0\leq t\leq T$,
\begin{align}\label{eq_duality_inequality}
\E^\Q[\bar{U}(Z_T , 1)|Z_t=z] 
\leq \E^\Q[\tilde{U}(Y_T, 1) |Y_t=y]+yz.
\end{align}
Define the dual value function $\tilde{g} \colon [0,T] \times (0, \infty) \to \R$ by
\begin{align}\label{eq_value_new_dual}
\tilde{g}(t, y) \defeq \inf_{\xi \in \HH^2_t(\R)}\E^\Q\left[ \tilde{U}\left(Y_T, 1\right) \middle| Y_t = y\right].
\end{align}
The HJB equation associated to the dual problem (\ref{eq_value_new_dual}) is 
$$ \partial_t \tilde{g} + \inf_{ \xi \in \R} \left\{ -y \partial_y \tilde{g} \left(\bar{\alpha} + b(1 - \rho^2)\xi\right) + \half y^2 \partial_{yy} \tilde{g}\left( \bar{\theta}^2 + \xi^2 (1 - \rho ^2)\right)\right\}=0. 
$$
The candidate dual optimal control is given by,  for $(t, y) \in [0,T] \times [0, \infty)$,
\[\xi^*(t, y) =  \frac{  \partial_y \tilde{g}(t, y) b}{  y \partial_{yy} \tilde{g}(t, y)},\]
which leads to  a nonlinear PDE
$$
  \partial_t \tilde{g}   - \bar{\alpha} y \partial_y \tilde{g}    -\half  (1 - \rho ^ 2) \frac{ (b\partial_y \tilde{g} )^2}{\partial_{yy} \tilde{g}}  + \half y ^2 \partial_{yy}\tilde{g} \bar{\theta}^2=0
$$
with the terminal condition $\tilde g(T,y)=\tilde U(y,1)$.

In the case $|\rho| = 1$, the HJB equation is independent of $\xi$, the dual value function is simply given by $\tilde{g}(t, y) = \E^\Q\left[ \tilde{U}\left(Y_T, 1\right)|Y_t=y\right]$ and can be computed explicitly,  where $Y_T$ is a lognormal random variable. 

Maximising the left side of (\ref{eq_duality_inequality}) over $\pi$ and minimising  the right side of (\ref{eq_duality_inequality}) over $\xi$ and $y$ yields the weak duality relation
\[g(t, z) \leq \bar{g}(t, z) \leq \inf_{y > 0} \left\{\tilde{g}(t, y)  + zy\right\},\]
for all $t \in [0,T]$ and $z \in (0, \infty)$. Since we mainly focus on numerical computation of value functions with deep neural network (DNN), we assume in the rest of the paper that $g, \bar g, \tilde g$ are $C^{1,2}$, the candidate optimal controls are well defined, and all satisfy required integrability conditions when substituted by underlying controlled state processes. We can then characterize and compute them with the corresponding HJB equations. 

\subsection{Dual Value Function and Concavification Principle in Complete Market Case} \label{sec_solution}

In this section we take $|\rho| = 1$, and evaluate the function $\tilde{g}$ given by (\ref{eq_value_new_dual}) explicitly. We use this explicit representation to solve the primal concavified problem and show that its solution coincides with the non-concave primal problem. We consider a general constant reference $\bar{R} \in [0, \infty)$, taking $\bar{R} = 1$ recovers (\ref{eq_value_new_dual}) and $\bar{R} = 0$ recovers the standard utility maximisation problem.

The dynamics of $Y$ in the complete market setting are given by
\[dY_t  = - \bar{\alpha} Y_t dt - \bar{\theta} Y_t d{\color{blue}\tilde{W}_t}, \quad 0 \leq t \leq T, \]
where $\bar{\alpha} = \alpha - a + \rho b \theta$ and $\bar{\theta} = \theta + \rho b (p - 1)$.
Since $Y$ follows a geometric Brownian motion, after some lengthy calculation, we get that,  for $(t, y) \in [0,T] \times (0,\infty)$,
\begin{eqnarray*}
\tilde{g}(t, y) 
& =& \E^{\color{blue}\Q}\left[\tilde{U}(Y_T, \bar{R}) \middle| Y_t = y\right] \\
& =& y^{\frac{p}{p-1}}p^\frac{p}{1-p}(1 - p) \exp\left(-{\frac{p}{p-1}}\left(\bar{\alpha} - \frac{|\bar{\theta}|^2}{2 (p - 1)}\right) \tau \right) \Phi\left(  -k(\tau, y) - \frac{p}{p-1} \bar{\theta} \sqrt{\tau}\right) \\
&& - \bar{R} y \exp(-\bar{\alpha} \tau) \Phi(-k(\tau, y)-\bar{\theta} \sqrt{\tau}) + u_0 \Phi(k(\tau, y))\\
\partial_y \tilde{g}(t, y) 
&=&  - y^{\frac{1}{p-1}}p^\frac{1}{1-p} \exp\left(-{\frac{p}{p-1}}\left(\bar{\alpha} - \frac{|\bar{\theta}|^2}{2 (p - 1)}\right) \tau \right) \Phi\left(  -k(\tau, y) - \frac{p}{p-1} \bar{\theta} \sqrt{\tau}\right) \\
&& - \bar{R}  \exp(-\bar{\alpha} \tau) \Phi(-k(\tau, y)-\bar{\theta} \sqrt{\tau}) \\
\partial_{yy} \tilde{g}(t, y)
& =& y^{\frac{2 - p}{p-1}}p^{\frac{1}{1-p}}\frac{1}{1-p} \exp\left(-{\frac{p}{p-1}}\left(\bar{\alpha} - \frac{|\bar{\theta}|^2}{2 (p - 1)}\right) \tau \right) \Phi\left(  -k(\tau, y) - \frac{p}{p-1} \bar{\theta} \sqrt{\tau}\right) \\
&& + \phi\left( k(\tau, y) \right) \frac{\hat{z} \hat{u}}{\bar{\theta} \sqrt{\tau} y^2},
\end{eqnarray*}
where $\Phi$ is the cumulative distribution function of a standard normal variable, $\phi(x)=\Phi'(x)$, $\tau = T - t$, $\hat{z} = {\color{blue} \eta}(\bar{R})$, $\hat{u} = U'(\hat{z} - \bar{R}) = p\left(\hat{z} - \bar{R}\right)^{p-1}$, $u_0 = -U_2(\bar{R}) = - K \bar{R}^p$, and 
\[k(\tau, y) \defeq \frac{\log(y) - \log(\hat{u}) - (\bar{\alpha} + \half |\bar{\theta}|^2) \tau}{\bar{\theta} \sqrt{\tau}}.\]

Note that $\tilde{g}\in C^{1,2}$, $\partial_{yy} \tilde{g}(t, y) > 0$, $\lim_{y \to 0} \partial_y \tilde{g}(t, y)  = -\infty$ and  $\lim_{y \to \infty} \partial_y \tilde{g}(t, y) = 0$ for all $(t, y) \in [0,T) \times (0, \infty)$, so $\tilde g(t,y)$ is a strictly convex decreasing function in $y$ for all $t\in [0,T)$. 
 The optimal state process for the concavified problem is given by 
\begin{align} \label{eq_dual_to_primal}
Z^*_t = -\partial_y \tilde{g}(t, Y^*_t), \quad 0 \leq t \leq T,
\end{align}
where $Y^*$ satisfies the dual state dynamics with the initial value determined from the equation $z_0 = -\partial_y \tilde{g}(0, Y^*_0)$, where $z_0 = \frac{x_0}{r_0}$ is the initial scaled wealth. 

The next result shows that $U(Z^*_T - \bar{R}) = \bar{U}(Z^*_T, \bar{R} )$ and therefore $Z^*$ is optimal for the original problem.
\begin{lemma} \label{lem_conv_principle} The optimal state process $Z^*$ given by (\ref{eq_dual_to_primal}) satisfies
$Z^*_T \in \{0\} \cup [\hat{z}, \infty)$ almost surely.
\end{lemma}
\begin{proof}
In the limit $\tau \to 0$ (or $t \to T$), we have 
$k(\tau, y) \to \infty$ if $y > \hat{u}$ and $-\infty$ if $y < \hat{u}$. If $y > \hat{u}$ then $ \Phi(-k(\tau, y)) \to 0$, hence $Z^*_t \to 0$. If $y < \hat{u}$ then $ \Phi(-k(\tau, y)) \to 1$, hence, noting that $y \mapsto y^{\frac{1}{p-1}}$ is a decreasing function,
\[Z^*_t \to   y^{\frac{1}{p - 1}}p^{\frac{1}{1 - p}} + \bar{R} > p^{\frac{1}{1 - p}} \left(p (\hat{z} - \bar{R})^{p - 1}\right)^{\frac{1}{p - 1}} + \bar{R} = \hat{z}. \]
If $y = \hat{u}$ then $\Phi(-k(\tau, y)) \to \Phi(0) = \half$ and $Z^*_t \to  \half \hat{z} \notin  \{0\} \cup [\hat{z}, \infty)$, but $\p(Y^*_T = \hat{u}) = 0$.
\end{proof}

\subsection{PINN Method}

We begin by describing the PINN method for general second order PDEs, with dimensional state in the region $(0, \infty)^d$. Let $F \colon [0,T) \times (0, \infty)^d \times \R \times \R \times \R^d \times \R^{d \times d} \to \R$, $G \colon (0, \infty)^d \to \R$ and $H \colon [0,T) \times \R^d \to \R$ be functions. Consider the PDE \begin{align*}
F(t, x, u, \partial_t u, \partial_x u, \partial_{xx}u) & = 0, \quad t \in [0,T), X \in (0, \infty)^d, \\
u(T, X) & = G(X) , \quad X \in (0, \infty)^d, \\
u(t, X) & = H(t, X), \quad t \in [0,T), X \in \partial (0,\infty)^d.
\end{align*}
where $\partial (0,\infty)^d \defeq \{(x_i)_{i = 1}^d \in (0, \infty)^d \colon x_i = 0 \text{ for some } i = 1, \ldots, d\}$ denotes the finite boundary of $(0,\infty)^d$.
Let $(t, X) \in [0,T] \times \R^d$. We solve the PDE using the PINN method. The PDE solution is approximated by a neural network
\begin{align*}
u^{\Theta}(t, X) & \defeq \Norm(t, X; \Theta)
\end{align*}
based on some  parameter set $\Theta$. The function $\Norm$ is defined as a series of alternating compositions of linear with a nonlinear activation function $h \colon \R \to \R$ (acting element-wise)
\[\Norm(t, X; \Theta) = f_M \circ h \circ f_{M-1} \cdots \circ h \circ f_0((t, X)),\]
for some fixed $M \in \N$ (representing the number of `layers' of the network), where $f_i(z) = A_i z + b_i$ for some matrix $A_i$ and vector $b_i$ of appropriate (fixed) dimension, and we view $(t, X)$ as a single vector in $\R^{1 + d}$. The dimension of the final layer is 1, but the middle layers have some pre-determined higher dimension. The larger the dimensions of $b_i$ for $i = 0, \ldots, M - 1$, the better approximation capabilities of the neural network \cite{15}, but also the more variables in the network, leading to longer training times. For the activation function, in practise we take $h(z) = \tanh(z)$. The parameter set $\Theta$ is made up of the individual elements of $(A_i, b_i)_{i = 0}^M$, which are optimised by the PINN algorithm, using a loss function defined in the sequel. At step 0 we begin with a parameter set $\Theta_0$ with elements initialised randomly.

We choose a grid size $\mathcal{X} \defeq (X_{\min}, X_{\max})^d \in (0, \infty)^d$, then uniformly sample three batches: the collocation data within the time horizon, given by  $(t_j, X_j) \in [0,T) \times \mathcal{X} $ for $j = 1, \ldots, \mathcal{M}_c$,  the boundary data $(T, X^b_j) \in \{T\} \times \mathcal{X}$ for $j = 1, \ldots, \mathcal{M}_b$ and the zero state data $(t^s_j, X^s_j) \in [0,T] \times \partial(X_{\min}, X_{\max})^d$ for $j = 1, \ldots, \mathcal{M}_s$.
Using these sample data we define the loss function, accounting for the PDE and the terminal and boundary conditions, as a combination of sample averages
\begin{eqnarray*}
L(\Theta) & =& \frac{1}{\mathcal{M}_c}\sum_{ j = 1}^{\mathcal{M}_c} \left|F\left(t_j, X_j, u^\Theta\left(t_j, X_j\right), \partial_t u^\Theta\left(t_j, X_j\right), \partial_x u^\Theta\left(t_j, X_j\right), \partial_{xx}u^\Theta\left(t_j, X_j\right)\right)\right|^2 \\
&& + \frac{1}{\mathcal{M}_b}\sum_{ j = 1}^{\mathcal{M}_b}  \left|u^{\Theta}\left(T, X^b_j\right) - G\left(X^b_j\right)\right|^2  + \frac{1}{\mathcal{M}_s}\sum_{ j = 1}^{\mathcal{M}_s}  \left|u^{\Theta}\left(t^s_j,X^s_j\right) - H(t^s_j)\right|^2.
\end{eqnarray*}
In this loss function we assume that $X_{\min} \approx 0$. We may then perform the following update (acting element-wise on the parameters of the neural networks)
\begin{align*}
\Theta_{k+1} & = \Theta_{k} - \epsilon \partial_{\Theta}L(\Theta_k), \quad k \in \N,
\end{align*}
given some learning rate $\epsilon > 0$. We repeat this step until convergence - either a fixed number of iteration steps, or a suitably low loss is attained.

We can use the PINN to solve the primal and dual HJB equations. 
For example, for the primal problem, upon convergence of the algorithm, with some outputted parameter set $\Theta$ and neural network function denoted by $g^\Theta$, we rescale and output the following approximation of $v$ at time 0. Given some range $(Z_{\min}, Z_{\max}) \subset \R$ of $z$ values used in the PINN, and some output size $\mathcal{B} \in \N$, we define
\begin{align} \begin{split} \label{eq_output_pinn1}
Z_i &= \left(Z_{\min} + \frac{i}{\mathcal{B}} (Z_{\max}-Z_{\min})\right), \\
X_i &= r_0 Z_i, \\
v_i &= r_0^p H_{T} g^\Theta(0, Z_i), \\
\pi_i & =  - \frac{\partial_z g^\Theta(0, Z_i)  \bar{\theta}  }{\partial_{zz} g^\Theta(0, Z_i)Z_i \sigma} +   \frac{\rho b}{\sigma}, 
\end{split}
\end{align}
for $i = 1, \ldots, {\mathcal{B}}$, where $H$ is defined in (\ref{eq_value_conversion}). 

The dual problem can be solved similarly. 
We do not optimise the dual start value $y$, which is needed to solve the primal problem. However we end up with a range of values of $\partial_y \tilde{g}(t, y)$ which correlate with the $z$. When there is no duality gap, we have the following relation between the dual and primal problems for $t \in [0,T]$
\begin{align} 
\begin{split}
\bar{g}(t, Z^*_t) & = \tilde{g}(t, Y^*_t) + Y^*_t Z^*_t, \\
Z^*_t & = - \partial_y \tilde{g}(t, Y^*_t), \\
\partial_z \bar{g}(t, Z^*_t) & = Y^*_t, \\
\partial_{zz} \bar{g}(t, Z^*_t) & = \frac{-1}{\partial_{yy} \tilde{g}(t, Y^*_t)}.
\end{split}
\label{eq_duality_relations}
\end{align}
\begin{remark}
In practice we do not know if these are true, especially in the incomplete market case, where the value function may not be differentiable and there may be a duality gap. We can still evaluate these terms though, as the neural network function is  continuously differentiable, and compare them to their primal counterparts to see if there is a gap.
\end{remark}
Upon convergence of the algorithm, with some outputted parameter set $\tilde{\Theta}$ and neural network function $\tilde{g}^{\tilde{\Theta}}$, we rescale and output the following approximation of $v$ at time 0. Given some range $(Y_{\min}, Y_{\max}) \subset \R$ of $y$ values used in the PINN, and some output size ${\mathcal{B}} \in \N$, we define
\begin{align}  \begin{split} \label{eq_output_pinn2}
Y_i &= \left(Y_{\min} + \frac{i}{\mathcal{B}} (Y_{\max}-Y_{\min})\right), \\
X^d_i &= - r_0 \partial_y \tilde{g}^{\tilde{\Theta}}(0, Y_i), \\
v^d_i &= r_0^p H_{T} \left(\tilde{g}^{\tilde{\Theta}}(0, Y_i) - Y_i \partial_y \tilde{g}^{\tilde{\Theta}}(0, Y_i)\right) , \\
\pi^d_i & =  - \frac{Y_i  \bar{\theta} \partial_{yy} \tilde{g}^{\tilde{\Theta}}(0, Y_i) }{\partial_y \tilde{g}^{\tilde{\Theta}}(0, Y_i) \sigma} +   \frac{\rho b}{\sigma}, 
\end{split}
\end{align}
for $i = 1, \ldots, {\mathcal{B}}$, where $H$ is defined in (\ref{eq_value_conversion}). 

\section{General Utility} \label{sec_4}

We return to the original problem for general utilities $U_1, U_2$, given by (\ref{eq_value}).  
Since we mainly focus on numerical computation of value functions with DNN, we again assume all regularity and integrability conditions are satisfied for all functions whenever needed and all equations are well defined with existence and uniqueness of solutions.  We can then characterize value functions, optimal controls, etc. with the  DPP or the SMP and solve the corresponding HJB equation or the adjoint BSDE with DNN. 

\subsection{PINN for General Utility}

The value function  $v$ is a classical solution to the HJB equation 
\begin{align} \begin{split} \label{eq_hjb_general}
0 & = \partial_t v + \sup_{\pi \in \R} \left\{ x(\alpha + \pi \sigma \theta) \partial_x v + \half x^2 \pi^2 \sigma^2 \partial_{xx} v + a r \partial_r v + \half b^2 r^2 \partial_{rr}v + \rho \sigma \pi b r x \partial_{xr} v \right\}, \\
& =  \partial_t v +  x\alpha \partial_x v + a r \partial_r v + \half b^2 r^2 \partial_{rr}v  - \half \frac{\left(\theta \partial_x v + \rho b r  \partial_{xr} v\right)^2}{ \partial_{xx} v }
\end{split}
\end{align}
with the terminal condition $v(T,x, r) = U(x - r)$. The optimiser $\pi^*$ is given by, for $(t, x, r) \in [0,T] \times (0,\infty)^2$,
\begin{align} \label{eq_hjb_control}
\pi^*(t, x, r) = -  \frac{\theta \partial_x v(t, x, r) + \rho b r  \partial_{xr} v(t, x, r)}{ \sigma x \partial_{xx} v(t, x, r) }.
\end{align}
Similarly for the concave problem (\ref{eq_value_conc}) the function $\bar{v}$ satisfies (\ref{eq_hjb_general}) with the terminal condition $\bar{v}(T,x, r) =\bar {U}(x , r)$. We may define the dual problem here, similar to (\ref{eq_value_new_dual}), for $(t, x, r) \in [0,T] \times [0,\infty)^2$,
\[\tilde{v}(t, y, r) = \inf_{\xi \in \HH^2_t(\R)}\E\left[\tilde{U}(Y_T, R_T) \middle| Y_t = y, R_t = r\right]\]
subject to
\begin{align*}
dY_t = -Y_t \alpha dt - Y_t \left( \theta - \rho \xi_t \right) dW_t - Y_t \xi_t dW_t^R, \quad 0 \leq t \leq T.
\end{align*}
The HJB equation is 
\begin{align*}
0  = \partial_t \tilde{v} + \inf_{ \xi} \bigg\{  &-y \alpha \partial_y \tilde{v} + \half y^2  \left(\theta^2 + \xi^2(1 - \rho^2)\right) \partial_{yy} \tilde{v}\\
&  + a r \partial_r v + \half b^2 r^2 \partial_{rr}\tilde{v} -  r b y \left( \rho  \theta + \xi(1 - \rho^2) \right) \partial_{yr} \tilde{v} \bigg\},  \\
 =  \partial_t \tilde{v}  -y \alpha \partial_y \tilde{v} &+ \half y^2  \theta^2 \partial_{yy} \tilde{v}  + a r \partial_r \tilde{v} + \half b^2 r^2 \partial_{rr}\tilde{v} -  r b y  \rho  \theta \partial_{yr} \tilde{v} 
 - \half \frac{(r b\partial_{yr} \tilde{v})^2(1 - \rho^2)}{\partial_{yy} \tilde{v}} 
\end{align*}
with the terminal condition $\tilde{v}(T,y, r) =\tilde {U}(y, r)$. 

Again, we have three PDEs that the PINN method may be applied to. The difference here is that the state process becomes multidimensional, with $d = 2$. The output of these algorithms is  the same as (\ref{eq_output_pinn1}) and (\ref{eq_output_pinn2}), applied at a range of $r$ values. We can use the boundary conditions, for $(t, r) \in [0, T] \times [0, \infty)$, 
\begin{align*}
v(t, 0, r) =  \bar{v}(t, 0, r) = \E\left[-U_2(R_T)\right]. 
\end{align*}
Additionally, when $r = 0$ we can use the standard utility maximisation result, which is a special case of the solution in Section \ref{sec_solution}, given by, for $(t, x) \in [0, T] \times [0, \infty)$,
\[v(t, x, 0) =e^{p\left(\alpha + \frac{\theta ^ 2}{2 (1 - p)}\right)   (T - t)} x^p. \]

\subsection{SMP for General Utility}

As an alternative to the HJB method, which requires a Markovian structure, we may apply the SMP instead. In this method we treat the reference $R$ as a random variable (that we can simulate) and treat the utility as a random function, taking in $(\omega, x) \in \Omega \times (0, \infty)$ as input. 
To apply the SMP as it is given in the literature, for example \cite{21}, we assume a complete market here, so $R_T$ is replicable. We also require the terminal utility to be concave, so we apply the SMP to the concave problem. 

For this section we change the problem slightly, defining $\Pi_t = \pi_tX_t$ and letting this be our control variable. This means that the control represents the amount invested in the stock, rather than the proportion of wealth invested. The state process dynamics becomes
\begin{align} \label{eq_state_amount}
dX_t = \left(rX_t + \sigma \theta \Pi_t \right)dt + \sigma \Pi_t dW_t, \quad 0 \leq t \leq T.
\end{align}
In this case, $X$ is no longer a geometric Brownian motion and may be negative. We therefore require $\Pi$ to be in the admissible control set, defined by, for $(t, x) \in [0,T] \times [0,\infty)$,
\begin{align*} 
\A(t, x) = \left\{ \Pi \in \HH^2_t(\R) \middle| X^{t, x, \Pi}_s \geq 0 \text{ a.s. for all } s \in [t,T] \right\},
\end{align*}
where $X^{t, x, \Pi}$ satisfies (\ref{eq_state_amount}) with $X^{t, x, \Pi}_t = x$. 
Now we define the random utility $U^R \colon \Omega \times [0,\infty) \to \R$ by $U^R(\omega, x) \defeq \bar{U}(x, R_T(\omega))$, and the problem becomes, for $(t, x, r) \in [0,T] \times [0,\infty)^2$,
\[v(t, x, r) = \sup_{\Pi \in \A(t, x)} \E\left[ U^R(X_T)\middle| X_t = x, R_t = r\right].\]
If this expectation cannot be determined explicitly, we evaluate it numerically. 
Let $(\hat{X}, \hat{\Pi})$ be an optimal state-control pair. Define the adjoint BSDE by, $t \in [0,T]$,
\begin{align}\begin{split}
dp_t & = -\alpha p_t dt + q_t dW_t ,   \quad 
 p_T  = -\partial_x \bar{U}(\hat{X}_T, R_T).
 \end{split} \label{eq_adjoint}
\end{align}
The SMP (see \cite{21}[Theorem 3.2]) states that  the pair ($\hat{X}, \hat{\Pi}$) is optimal if and only if for all $\Pi \in \R$, $t \in [0,T]$,
\begin{align*}
 \sigma(\theta  p_t +  q_t)(\hat{\Pi}_t - \Pi) \geq 0,
 \end{align*}
which implies  $q_t= - \theta  p_t$, 
 and $p$ has a closed-from expression  
 $p_t  = p_0 e^{-(\alpha + \half \theta^2)t - \theta W_t }$ for $0 \leq t \leq T$.
 We need to choose the initial point $p_0$ and the control  $\hat{\Pi}$ such that the terminal condition in (\ref{eq_adjoint}) is satisfied. 
To this end,   we define two neural networks as follows. For $(t, x) \in [0,T] \times [0,\infty)$,
\begin{align*}
\Pi^{\Theta_\Pi}(t, x) & = \Norm(t, x;\Theta_\Pi), \quad
p^{\Theta_p}(x)  = \Norm(x; \Theta_p),
\end{align*}
where $\Theta_\Pi$ and $\Theta_p$ denote the parameters of two neural networks, with some fixed number of hidden layers and nodes per layer.
Let $t_i = i\Delta t$, $i=0,\ldots,N$, $\Delta t=T/N$.
We choose a grid size $\mathcal{X} \defeq (X_{\min}, X_{\max}) \in (0, \infty)$.
We then define the loss function, for a given batch of $X^j_0 \in \mathcal{X}, \Delta W^j_i \sim \text{Norm}(0,1)$, $j = 1, \ldots, \M$, $i = 1, \ldots, N - 1$ as 
\begin{align} \label{eq_loss_smp}
\LL(\Theta_\Pi, \Theta_p) = \frac{1}{\M}\sum_{j = 1}^ {\M} \left|e^{-(\alpha + \half \theta^2)T - \theta W^j_T}p^{\Theta_p}(X^j_0) +  \partial_x\bar{U}(X^j_N, R_T)\right|^2 - \frac{1}{\M}\sum_{j = 1}^ {\M} \bar{U}(X^j_N, R_T),
\end{align}
where
\begin{align*}
X_{i+1}^j &= X_i^j + \left(\alpha X^j_i + \Pi^{\Theta_\Pi}(t_i, X^j_i) \theta \sigma \right) \Delta t+ \Pi^{\Theta_\Pi}(t_i, X^j_i) \sigma \sqrt{\Delta t} \Delta W^j_i, \\
W^j_N &=  \sqrt{\Delta t} \sum_{i = 0}^{N-1} \Delta W^j_i.
\end{align*}
We then iteratively perform the gradient descent update, for $k \in \N$,
\begin{align*}
\Theta_\Pi^{k+1} & = \Theta_\Pi^k - \lambda \partial_\Pi \LL(\Theta_\Pi^k, \Theta_p^k) \\
\Theta_p^{k+1} & = \Theta_p^k - \lambda \partial_p \LL(\Theta_\Pi^k, \Theta_p^k), 
\end{align*}
where $\partial_\Pi \LL$ and $\partial_p \LL$ denote the partial derivatives of $\LL$ with respect to the first and second element respectively, and $\lambda > 0$ is some learning rate. We interpret this update as the operation occurring to each individual element of the sets $\Theta_\Pi$ and $\Theta_p$. Note that while the first element of (\ref{eq_loss_smp}) should go to 0, the second term evaluates the gains function so will not be 0, but should converge to the value function as $k, \M \to \infty$. If we only had the second term in the loss function this method would not involve the adjoint equation and reduce to the method of \cite{11}.

Our $\M$ is typically small (in practise we use $\M = 500$), so the loss function does not provide a precise approximation of the value function. To output the value function for this method, we compute $ \E\left[ \bar{U}(X_T, R_T)\right]$ where $X$ is driven by $\Pi^{\Theta_{\Pi}}$ at time 0 using Monte Carlo with a much higher sample size, with the converged control parameter set $\Theta_\Pi$.

Upon convergence of the algorithm, with some outputted parameter set $\Theta_\Pi$ and neural network function $\Pi^{\Theta_\Pi}$, we output the following approximation of $v$ at time 0. Given some output size ${\mathcal{B}} \in \N$ we define
\begin{align}  \begin{split} \label{eq_output_smp}
X^s_i &= \left(X_{\min} + \frac{i}{\mathcal{B}} (X_{\max}-X_{\min})\right), \\
v^s_i &= \E^s\left[ \bar{U}(X^{i, \Pi^{\Theta_\Pi}}_T, R_T) \right], \\
\Pi^s_i & =  \Pi^{\Theta_\Pi}(0, X_i), 
\end{split}
\end{align}
for $i = 1, \ldots, {\mathcal{B}}$, where $\E^s$ here indicates a sample mean of sufficiently large sample size, and $X^{i, \Pi^{\Theta_\Pi}}$ is the wealth process satisfying (\ref{eq_state_amount}) starting at $X_i$ and driven by $\Pi^{\Theta_\Pi}$.

\section{Numerical Examples} \label{sec_5}

We take the problem coefficients $\alpha = 0.05$, $\sigma = 0.2$, $\theta = 0.5$, $K = 0.5$, $T = 0.5$, $a = 0.03$, $b = 0.1$. When we consider the power utility, we take $p = 0.5$.
We take the neural network to have 4 layers, with 2 hidden layers consisting of 50 nodes, and a tanh activation function. We sample points using $\mathcal{Z} = \mathcal{X} = (0.05, 5)$ and $\mathcal{Y} = (0.25, 1.0)$ for the SMP, primal and dual PINN methods respectively. We sample $R_0$ using the same range for the general PINN methods, and fixed  to a constant for the SMP method (either 0.5, 1.0, or 5.0 for our results). We take sample sizes of $\mathcal{M}_c = 1000$, $\mathcal{M}_b = \mathcal{M}_s = 100$, $\M = 500$, and run Monte Carlo simulations with $N = 100$. We use 20000 iteration steps for the PINN methods and 1000 steps for the SMP method, with a learning rate of $\epsilon = 0.001$ and $\epsilon = 0.01$ for the PINN and SMP methods respectively. If the loss function is below $5 \times 10^{-5}$, we terminate the algorithm early. We generate equally spaced points along our grids, evaluate the neural network functions with converged parameter sets, and compare this to a plot of the solutions (red), evaluated (also numerically) using the dual solutions in   Section \ref{sec_solution}. 

We run the algorithms using a PC with a Intel Core i9 9900K and 16.0GB DDR4 RAM, without GPU acceleration. The runtime for the algorithms for this setup is roughly 1 minute for the scaled PINN methods and 2 minutes for the general PINN methods, though the primal concavified PINN method often terminates early due to low loss. The SMP method takes about 1 minute. We run the algorithms 10 times and present the average results. 

See \texttt{https://github.com/Ashley-Davey/ML-for-S-Utility}
for the code to reproduce the results in this section.

\subsection{Converting from Scaled Problem to Regular Problem}

For the problem with scalable utility, we use the PINN methods to find the function $g$, and this function is found in the solution as well. To convert to the value function $v$ we use (\ref{eq_value_conversion}).
To recover the optimal control in (\ref{eq_hjb_control}) we use  (\ref{eq_opt_control_reduced}).

\subsection{Converting from Dual Problem to Primal Problem}

We solve for the function $\tilde{v}$ using the dual problem. To convert to the primal problem we output the candidate primal value function as, for $(t, x, r) \in [0,T] \times [0,\infty)^2$,
\begin{align} \label{eq_dual_to_primal_func}
v^d(t, x, r) = \tilde{v}(t, y(t, x, r), r) + x y(t, x, r),
\end{align}
where $y(t, x, r)$ is the minimiser of $y \mapsto \tilde{v}(t, y, r) + x y$. If there is no duality gap, we have $v^d = v$ as the true value function. To recover the candidate optimal control representing (\ref{eq_hjb_control}) we note
\begin{align*}
\partial_x v^d(t, x, r) & = y(t, x, r), \\
\partial_{xx} v^d(t, x, r) & =  \frac{-1}{\partial_{yy} \tilde{v}(t, y(t, x, r), r)}, \\
\partial_{xr} v^d(t, x, r) & = \frac{- \partial_{yr} \tilde{v}(t, y(t, x, r), r)}{\partial_{yy} \tilde{v}(t, y(t, x, r), r)}.
\end{align*}
Therefore the candidate optimal control is given by
\begin{align*}
\pi^d(t, x, r) =   \frac{\theta y(t, x, r) \partial_{yy} \tilde{v}(t, y(t, x, r), r) - \rho b r  \partial_{yr}\tilde{v}(t, y(t, x, r), r)}{ \sigma x }.
\end{align*}

\subsection{Complete Market Problem, Scalable Utility}

For this example we take $\rho = 1$, $U_1(x) =  \sqrt{x}$, $U_2(x) = \half\sqrt{x}$. We compare our methods to the solution found using the dual method given in  Section \ref{sec_solution}.
 We apply 5 methods to this problem, the PINN method may be applied to the scaled problem, to the primal problem with both original (labelled as $g$) and concavified utility (labelled as $\bar{g}$), and to the dual problem. We plot $X$ against $\Pi = \pi X$ and $v$ as they are given in (\ref{eq_output_pinn1}) and (\ref{eq_output_pinn2}). We also apply the SMP algorithm, plotting $X$ against $\Pi$ and $v$ as they are given in (\ref{eq_output_smp}).

Figure \ref{fig_results_comp_scale} shows accuracy of the control and value function for all methods. There is variance in the control function due to the dependence of $\pi$ on the second derivative $\partial_{xx}v$, which is often approximated with less accuracy than the value function its self. 
In terms of control structure, we see a bump between 0 and the reference $R$, which matches our intuition that the agent is more risky here to ensure the wealth moves to either $0$ or $\eta(R_T)$ by the terminal time, so the concave utility of wealth agrees with the non-concave utility. As time approaches the terminal time, the control in this region goes to infinity.
Tables \ref{tab_results_comp_scale_1} and \ref{tab_results_comp_scale_2} provide
the value function approximations for all our methods with different initial values. 

\begin{remark}  \label{rem_nonconcave}
The non-concave method  performs poorly when the state $x$ is much lower than the reference $r$. This is likely due to the jump between the concave and non-concave utilities leading to a discontinuity in time of the value function at terminal time. Indeed, in the deterministic reference case with complete market, by \cite[Theorem 3.9]{1} the value function $v(t, x, r)$ will converge to $\bar{U}(x, r)$ as $t \to T$. This leads to a lack of continuity in the value function (it is only a viscosity solution to the HJB with non-concave terminal condition). We therefore cannot expect the value function to be replicated uniformly by a neural network function that is continuous, and equal to $U$ at terminal time.  
\end{remark}

\begin{figure}[H] 
\centering
\includegraphics[width=0.8\textwidth]{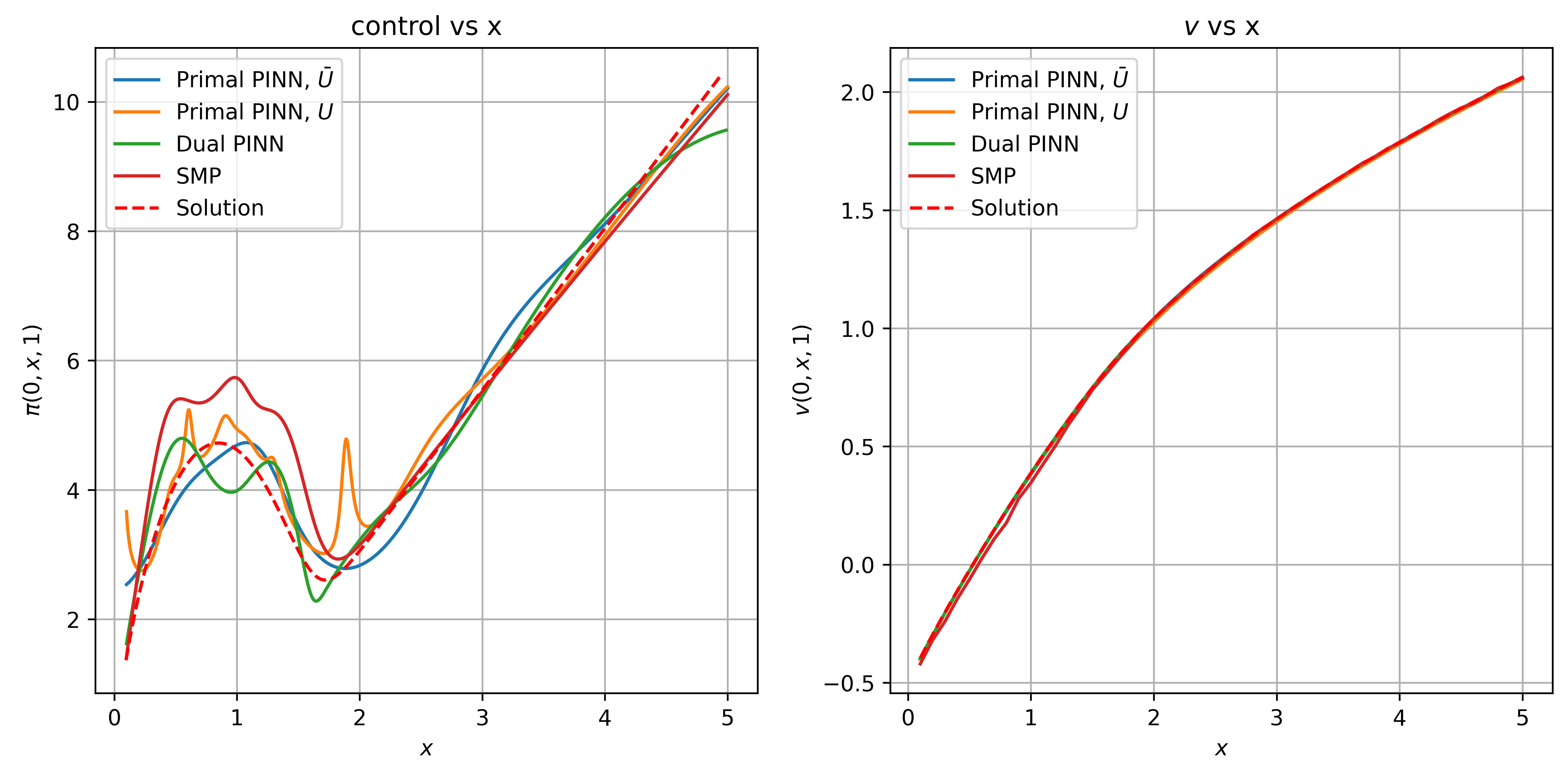}
\caption{The (amount invested) control and value functions $\Pi, v$ against state $x$, evaluated at $t = 0$ with complete market ($\rho = 1.0$), $r = 1.0$, and scalable utility. The red dotted line is the solution derived in Section \ref{sec_solution}.}
\label{fig_results_comp_scale}
\end{figure}
 

\begin{table}[H]
\centering
\begin{tabular}{|c|c|c|c|c|}\hline
(x, r) 	&	Primal	&	Primal Non-Conc	&	Dual	&	Solution	\\ \hline \hline 
(0.5, 1.0)	&	-0.0256	&	-0.0337	&	-0.0030	&	-0.0149	\\ \hline
(1.0, 1.0)	&	0.3844	&	0.3750	&	0.3940	&	0.3872	\\ \hline
(5.0, 1.0)	&	2.0618	&	2.0606	&	2.0610	&	2.0653	\\ \hline
(1.0, 0.5)	&	0.7316	&	0.7274	&	0.7246	&	0.7353	\\ \hline
(1.0, 5.0)	&	-0.6822	&	-0.6988	&	-0.6670	&	-0.6635	\\ \hline
\end{tabular}
\caption{The value functions $v$ against state $x$ and reference point $r$, evaluated at $t = 0$ with complete market ($\rho = 1.0$), and scalable utility. Results show the PINN-based algorithms that reduce dimensionality of the problem, solving the scaled problem with $z = \frac{x}{r}$, then rescaling.}
\label{tab_results_comp_scale_1}
\end{table}

\begin{table}[H]
\centering
\begin{tabular}{|c|c|c|c|c|c|}\hline
(x, r) 	&	\multicolumn{1}{|p{2cm}|}{\centering Primal \\ General} 	&	\multicolumn{1}{|p{2cm}|}{\centering Primal Non-Conc	 \\ General} &	\multicolumn{1}{|p{2cm}|}{\centering Dual \\ General} 	&	SMP	&	Solution	\\ \hline \hline 
(0.5, 1.0)	&	-0.0188	&	-0.0577	&	-0.0050	&	-0.0450	&	-0.0149	\\ \hline
(1.0, 1.0)	&	0.3793	&	0.3340	&	0.3866	&	0.3735	&	0.3872	\\ \hline
(5.0, 1.0)	&	2.0546	&	2.0631	&	2.0787	&	2.0678	&	2.0653	\\ \hline
(1.0, 0.5)	&	0.7335	&	0.7429	&	0.7320	&	0.7361	&	0.7353	\\ \hline
(1.0, 5.0)	&	-0.6758	&	-0.8197	&	-0.6551	&	-0.6885	&	-0.6635	\\ \hline    
\end{tabular}
\caption{The value functions $v$ against state $x$ and reference point $r$, evaluated at $t = 0$ with complete market ($\rho = 1.0$), and scalable utility. Results show the PINN and SMP algorithms that solve the original problem without scaling.}
\label{tab_results_comp_scale_2}
\end{table}

\subsection{Incomplete Market Problem}

For this example we take $\rho = 0$, $U_1(x) =  \sqrt{x}$, $U_2(x) = \half\sqrt{x}$. 
Figure \ref{fig_results_incomp_scale} shows accuracy of the control and value function for all methods. We do not have an analytical solution in this case to compare to. Compared to Figure \ref{fig_results_comp_scale}, we see little difference in the structure of the control, and concavity in the value function. 
The value function approximations for all our methods are provided in Tables \ref{tab_results_incomp_scale_1} and \ref{tab_results_incomp_scale_2}. We apply all the methods from the previous example, even though the SMP and concavificiation principles might not apply in the incomplete market case. However, the concave primal, and dual PINN methods appear to agree, suggesting the concavification principle still occurs. As the primal and dual values coincide, up to numerical errors, we conclude that the payoff $X_T = -\partial_y \tilde{U}(Y_T, R_T)$ has been replicated by the primal algorithm, where $Y_T$ is the optimal payoff for the dual problem. Given this, we conjecture that the concavification principle holds for this problem. These numerical results show robustness of our methods outside of problem classes that are well understood. 

\begin{figure}[H] 
\centering
\includegraphics[width=0.8\textwidth]{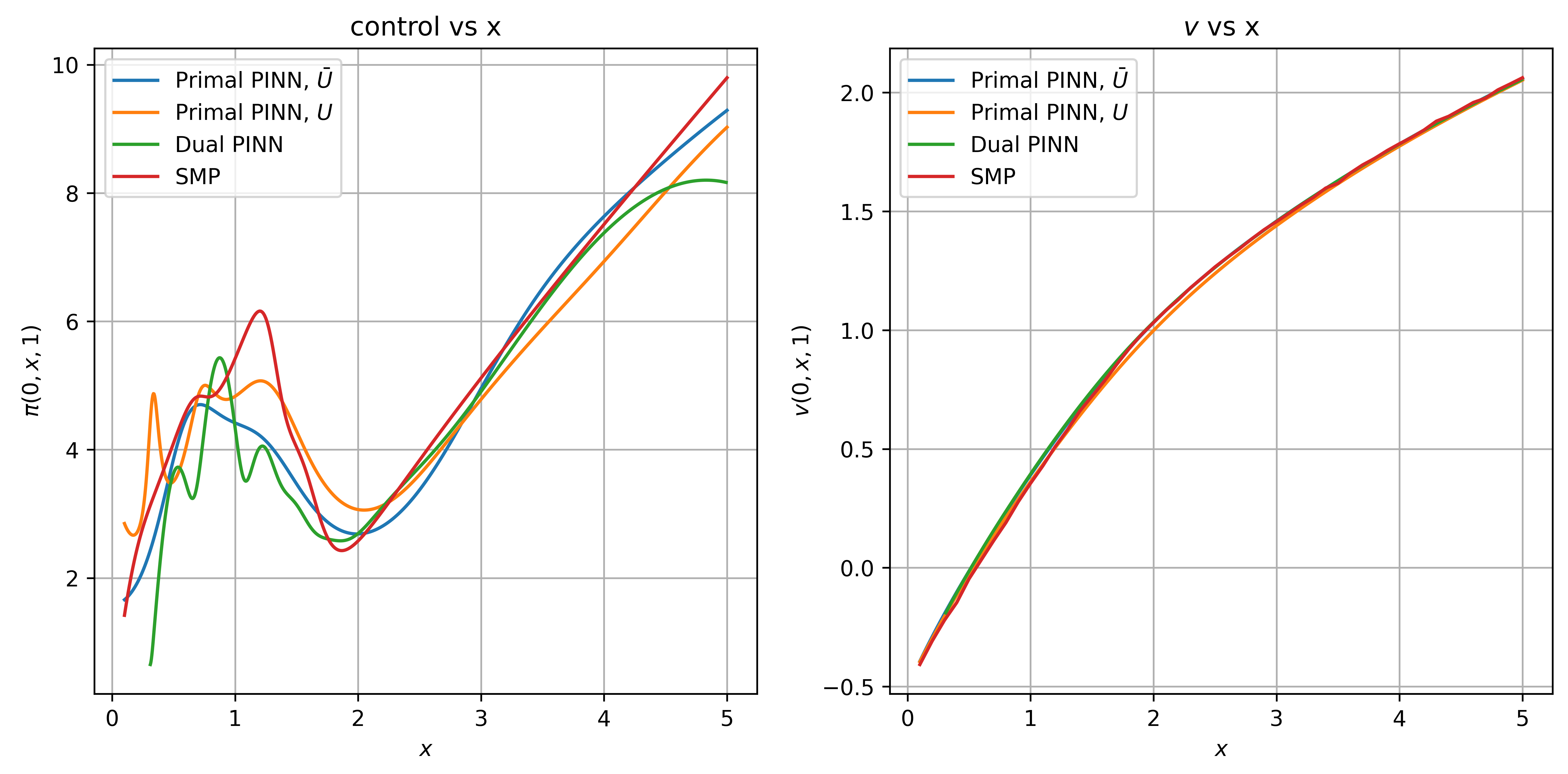}
\caption{The (amount invested) control and value functions $\Pi, v$ against state $x$, evaluated at $t = 0$ with incomplete market ($\rho = 0.0$), $r = 1.0$, and scalable utility.}
\label{fig_results_incomp_scale}
\end{figure}

\begin{table}[H]
\centering
\begin{tabular}{|c|c|c|c|}\hline
(x, r) 	&	Primal	&	Primal Non-Conc	&	Dual	\\ \hline \hline 
(0.5, 1.0)	&	-0.0144	&	-0.0391	&	0.0084	\\ \hline
(1.0, 1.0)	&	0.3931	&	0.3535	&	0.4045	\\ \hline
(5.0, 1.0)	&	2.0583	&	2.0521	&	2.0706	\\ \hline
(1.0, 0.5)	&	0.7288	&	0.7022	&	0.7211	\\ \hline
(1.0, 5.0)	&	-0.6650	&	-0.6888	&	-0.6162	\\ \hline  
\end{tabular}
\caption{The value functions $v$ against state $x$ and reference point $r$, evaluated at $t = 0$ with incomplete market ($\rho = 0.0$), and scalable utility. Results show the PINN-based algorithms that reduce dimensionality of the problem, solving the scaled problem with $z = \frac{x}{r}$, then rescaling.}
\label{tab_results_incomp_scale_1}
\end{table}

\begin{table}[H]
\centering
\begin{tabular}{|c|c|c|c|c|}\hline
(x, r) 	&	\multicolumn{1}{|p{2cm}|}{\centering Primal \\ General} 	&	\multicolumn{1}{|p{2cm}|}{\centering Primal Non-Conc	 \\ General} &	\multicolumn{1}{|p{2cm}|}{\centering Dual \\ General} 	&	SMP	\\ \hline \hline 
(0.5, 1.0)	&	-0.0066	&	-0.0635	&	0.0080	&	-0.0274	\\ \hline
(1.0, 1.0)	&	0.3895	&	0.3042	&	0.4087	&	0.3999	\\ \hline
(5.0, 1.0)	&	2.0517	&	2.0516	&	2.0610	&	2.0578	\\ \hline
(1.0, 0.5)	&	0.6930	&	0.6844	&	0.6987	&	0.7297	\\ \hline
(1.0, 5.0)	&	-0.6551	&	-0.7665	&	-0.6667	&	-0.6725	\\ \hline 
\end{tabular}
\caption{The value functions $v$ against state $x$ and reference point $r$, evaluated at $t = 0$ with incomplete market ($\rho = 0.0$), and scalable utility. Results show the PINN and SMP algorithms that solve the original problem without scaling.}
\label{tab_results_incomp_scale_2}
\end{table}

\subsection{General Utility}

For this example we take $\rho = 1$, $U_1(x) = \sqrt{x}$, $U_2(x) = \half\log(x + 1)$. 
The value function approximations for our methods are provided in Table \ref{tab_results_comp_general}. We cannot reduce this problem, so we apply the methods for solving general S-shaped utilities. In this example we compare the general PINN algorithms for the primal, primal concave, and dual problems, and the SMP. The results we see are familiar. In this setting the concavification principle holds, and we expect no duality gap. However, the non-concave PINN method does not perform well in the regions with $x < r$, see Remark \ref{rem_nonconcave}.

\begin{table}[H]
\centering
\begin{tabular}{|c|c|c|c|c|}\hline
(x, r) 	&	\multicolumn{1}{|p{2cm}|}{\centering Primal \\ General} 	&	\multicolumn{1}{|p{2cm}|}{\centering Primal Non-Conc	 \\ General} &	\multicolumn{1}{|p{2cm}|}{\centering Dual \\ General} 	&	SMP		\\ \hline \hline 
(0.5, 1.0)	&	0.0770	&	0.0327	&	0.0770	&	0.0496	\\ \hline
(1.0, 1.0)	&	0.4342	&	0.3638	&	0.4415	&	0.4230	\\ \hline
(5.0, 1.0)	&	2.0537	&	2.0654	&	2.0642	&	2.0613	\\ \hline
(1.0, 0.5)	&	0.7045	&	0.6184	&	0.7111	&	0.7274	\\ \hline
(1.0, 5.0)	&	-0.4916	&	-0.6596	&	-0.4742	&	-0.4829	\\ \hline
\end{tabular}
\caption{The value functions $v$ against state $x$ and reference point $r$, evaluated at $t = 0$ with complete market ($\rho = 1.0$), and general utility with no scaling property.}
\label{tab_results_comp_general}
\end{table}

\section{Conclusions} \label{sec_6}

In this paper we study the S-shaped utility maximisation problem, with a random benchmark that may not be replicable due to incompleteness. We introduce the deep neural network method for solving the HJB equation and the adjoint BSDE with the optimality condition from the SMP. We present numerical results for different solution methods (HJB equations for primal, concavified primal, dual value functions, SMP for adjoint BSDE, reduced and full problems) and  show that the DNN method is viable and robust in solving  nonlinear PDEs and BSDEs.  This paper has not addressed  
theoretical  convergence of DNN nor general concavification principle. Convergence analysis for the PINN method in particular is not available beyond linear PDEs, and our HJB equations are fully nonlinear PDEs. We leave these to future research.

\vskip 6mm
\noindent{\bf Acknowledgements}

\noindent
 The authors are  grateful to  the anonymous reviewer  whose comments have helped to improve the paper of the previous version.   Support by the Engineering and Physical Sciences Research Council (UK)  Grant (EP/V008331/1) is gratefully acknowledged.

\end{document}